\documentclass{lmcs}

\theoremstyle{plain} 

\RequirePackage{listings}
\usepackage{hyperref}
\usepackage[capitalise,noabbrev]{cleveref}

\def\ba{\begin{array}}
\def\ea{\end{array}}
\def\bda{\begin{displaymath}\ba}
\def\eda{\ea\end{displaymath}}

\newcommand{\thdSym}{\tau}
\newcommand{\thdSymB}{\upsilon}
\newcommand{\thread}[1]{\ensuremath{\thdSym_{#1}}}

\newcommand{\lockE}[1]{\mathtt{lock}(#1)}
\newcommand{\unlockE}[1]{\mathtt{unlock}(#1)}
\newcommand{\forkE}[1]{\mathtt{fork}(#1)}
\newcommand{\joinE}[1]{\mathtt{join}(#1)}

\newcommand{\eventE}[1]{e_{#1}}

\newcommand{\LK}[1]{m_{#1}}

\newcommand{\LKA}{\LK1}
\newcommand{\LKB}{\LK2}

\def\conc{\mathbin{\cdot}}

\newcommand{\cond}[1]{\textbf{#1}}

\newcommand{\TrLt}[1]{<_{tr}^{#1}}
\newcommand{\TrLtEq}[1]{\leq_{tr}^{#1}}
\newcommand{\crpLt}[1]{<_{crp}^{#1}}
\newcommand{\crpLtEq}[1]{\leq_{crp}^{#1}}

\newcommand{\evts}[1]{\mathit{evts}(#1)}
\newcommand{\thds}[1]{\mathit{thds}(#1)}
\newcommand{\crp}[1]{\mathit{crp}(#1)}

\newcommand{\pos}[2]{\mathit{pos}_{#1}(#2)}
\newcommand{\projThd}[2]{\mathit{proj}_{#1}(#2)}

\newcommand{\exit}[2]{\mathit{exit}_{#1}(#2)}

\newcommand{\exitH}[2]{\mathit{find}_{#1}(#2)}

\newcommand{\PerThreadLH}[3]{LH_{#1}^{#2}(#3)}
\newcommand{\PerThreadCS}[3]{CS_{#1}^{#2}(#3)}

\newcommand{\LH}[2]{LH_{#1}(#2)}
\newcommand{\CS}[2]{CS_{#1}(#2)}

\begin{document}

\bibliographystyle{alphaurl}

\title[Critical Sections Are Not Per-Thread]{Critical Sections Are Not Per-Thread: A Trace Semantics for Lock-Based Concurrency}

\author{Martin Sulzmann\lmcsorcid{0000-0002-8165-3403}}

\address{Karlsruhe University of Applied Sciences, Moltkestrasse 30, 76133 Karlsruhe, Germany}	
\email{martin.sulzmann@gmail.com}  

\begin{abstract}
  \noindent

Locks are a standard mechanism for synchronizing concurrent threads.
The standard lock set construction assumes that critical sections are confined to a single thread,
and therefore only accounts for locks acquired within that thread.
Traditional definitions of critical sections implicitly assume that protected events belong to the same thread. We demonstrate that this assumption does not hold for general C/Pthread executions.

Using a trace model that captures the essence of C/Pthread programs,
we give a trace-based characterization of critical sections
that does not impose a per-thread restriction.
As a result, critical sections may span multiple threads.
Such \emph{multi-thread} critical sections arise naturally in real programs
and close a semantic gap in the standard lock set construction.
\end{abstract}

\maketitle

\section{Introduction}

\begin{figure}

  \begin{lstlisting}[language=C,mathescape,escapechar=|]
#include <pthread.h>
pthread_mutex_t m1 = PTHREAD_MUTEX_INITIALIZER;
pthread_mutex_t m2 = PTHREAD_MUTEX_INITIALIZER;
void* thread_m1_m2(void*) {
  pthread_mutex_lock(&m1);      |\label{ln:tid2-acq-l1}|
  pthread_mutex_lock(&m2);      |\label{ln:tid2-acq-l2}|
  pthread_mutex_unlock(&m2);    |\label{ln:tid2-rel-l2}|
  pthread_mutex_unlock(&m1);    |\label{ln:tid2-rel-l1}|
return NULL;
}
void* thread_m1(void*) {
  pthread_mutex_lock(&m1);      |\label{ln:tid3-acq-l1}|
  pthread_mutex_unlock(&m1);
return NULL;
}
int main() {
  pthread_t tid2, tid3;
  pthread_create(&tid2, 0, thread_m1_m2, 0); |\label{ln:fork-tid2}|
  pthread_mutex_lock(&m2);                   |\label{ln:tid1-acq-l2}|
  pthread_create(&tid3, 0, thread_m1, 0);
  pthread_join(tid3, 0);                     |\label{ln:join-tid3}|
  pthread_mutex_unlock(&m2);                 |\label{ln:tid1-rel-l2}|
}
  \end{lstlisting}

\caption{C/Pthread program.}
  \label{fig:real}
\end{figure}

We consider concurrent programs that make use of threads and locks.
Many program analyses for deadlock and data race detection rely on the notion
of a critical section, which informally describes the region of execution protected by a lock.
In the literature, critical sections are typically defined as per-thread regions
delimited by a lock acquisition and the corresponding release in the same thread.
This definition implicitly assumes that the events protected by a lock belong
to the same thread as the lock acquisition.
However, this assumption does not hold for general executions of C/Pthread programs.

For example, consider the program in~\cref{fig:real}.
We assume that threads are numbered where
thread~1 is the thread in which the \texttt{main} function executes.
At line~\ref{ln:fork-tid2} another thread~2 is created.
Thread~2 acquires lock~\texttt{m1} and then acquires lock~\texttt{m2}
before releasing the locks.
The lock acquire operation at line~\ref{ln:tid2-acq-l2}
is part of the critical section with entry point at line~\ref{ln:tid2-acq-l1}
and exit point at line~\ref{ln:tid2-rel-l1}.

Entry and exit points of critical sections are pairs of lock and unlock operations
on the same lock variable that belong to the same thread.
Any event that must occur between the entry and exit points
in all execution traces belongs to the critical section.
The common assumption
is that ‘must occur between’ coincides with thread order,
and therefore that critical sections are per-thread.
This assumption is implicit and fundamental to the standard lock set construction~\cite{Dinning:1991:DAA:127695:122767},
yet it is rarely stated explicitly in the literature.

The above implicit assumptions do not apply to all C/Pthread programs
as shown by our example in~\cref{fig:real}.
The thread creation operation at line~\ref{ln:fork-tid2}
and the join operation at line~\ref{ln:join-tid3}
guarantee that the lock acquire operation at line~\ref{ln:tid3-acq-l1}
is always in between the critical section
represented by the entry point at line~\ref{ln:tid1-acq-l2}
and the exit point at line~\ref{ln:tid1-rel-l2}.
However, entry/exit points belong to a \emph{distinct} thread.
We conclude that the common notion of a critical section found in the literature
is semantically incomplete with respect to execution traces of C/Pthread programs.
Therefore, the standard lock set construction
that computes the set of locks that are held (acquired) but not yet released
at a certain point in the execution is incomplete as well.

\paragraph{Contributions}

The contributions of this paper are primarily conceptual:

\begin{itemize}
\item We identify an implicit per-thread assumption in the standard definition of critical sections and lock sets, and show that it is semantically incomplete
  with respect to a trace model that captures the essence of C/Pthread programs.

\item We give the first trace-based semantics of critical sections and lock sets that is complete for
  a minimal trace model capturing C/Pthread execution traces.
   The semantics captures \emph{multi-thread} critical sections in execution traces.

\end{itemize}

\paragraph{Outline}

\cref{sec:events-traces} introduces a minimal trace model for C/Pthread programs.
\cref{sec:cs-ls} reviews the commonly used per-thread definition
of a critical section and lock set and present our more general
trace-based characterization of a critical section and lock sets
\cref{sec:related-work} discusses earlier works that
rely on critical sections and lock sets for program analysis purposes.
\cref{sec:conclusion} concludes.

\section{Events and Traces}
\label{sec:events-traces}

We describe the semantics of programs in terms of execution traces.
A trace is a list of events reflecting the interleaved execution of a concurrent program
where each event is connected to some concurrency primitive.
We consider a trace model with events to represent
lock/unlock operations on locks and fork/join operations on threads.
This minimal model captures the essence of C/Pthread programs.

\begin{defi}[Events and Traces]
  \label{def:run-time-traces-events}
  \bda{@{}l@{}c@{}lll@{}c@{}llr@{}}
  e & {} ::= {} & (\alpha,t,op) & \text{(events)}
\\    \alpha,\beta,\delta & {} ::= {} & 1 \mid 2 \mid \ldots & \text{(unique event ids)}
\\ \thdSym,\thdSymB & {} ::= {} & \thread{1} \mid \thread{2} \mid \ldots & \text{(thread ids)}
    \\ op & {} ::= {} & \forkE{\thdSym} \mid \unlockE{m}
      \mid \joinE{\thdSym} \mid \lockE{m}  & \text{(operations)}
    \\ T & {} ::= {} & [] \mid e : T & \text{(traces)}
  \eda
\end{defi}

\noindent
We write $[o_1,\dots,o_n]$
for a list of objects as a shorthand of $o_1:\dots:o_n:[]$ and use the
operator  ``${}\conc{}$'' for list concatenation.

An event $e$ is represented by a triple $(\alpha, \thdSym,op)$
where
$\alpha$ is a unique event identifier,
$op$ is an operation, and
$\thdSym$ is the thread id in which the operation took place.
Via the unique event identifier we can unambiguously distinguish among
events that occur in the same thread and result from the same operation.
For brevity, we often omit the event identifier when denoting events
and write $e = (\thdSym,op)$ instead of $e = (\alpha,\thdSym,op)$.
We write~$\forkE{\thdSym}$ for the creation of a new thread with thread id~$\thdSym$
and~$\joinE{\thdSym}$ for a join with a thread with thread id~$\thdSym$.
We write~$\lockE{m}$ and~$\unlockE{m}$ for locking and unlocking of some lock~$m$.

\begin{figure}

\bda{|l|l|l|l|}
\hline  & \thread{1} & \thread{2} & \thread{3}\\ \hline
\eventE{1}  & \forkE{\thread{2}}&&\\
\eventE{2}  & &\lockE{\LKA}&\\
\eventE{3}  & &\lockE{\LKB}&\\
\eventE{4}  & &\unlockE{\LKB}&\\
\eventE{5}  & &\unlockE{\LKA}&\\
\eventE{6}  & \lockE{\LKA}&&\\
\eventE{7}  & \forkE{\thread{3}}&&\\
\eventE{8}  & &&\lockE{\LKB}\\
\eventE{9}  & &&\unlockE{\LKB}\\
\eventE{10}  & \joinE{\thread{3}}&&\\
\eventE{11}  & \unlockE{\LKA}&&\\

 \hline \eda{}
\caption{Trace $T_1$ resulting from execution of the program in \cref{fig:real}.}
 \label{fig:trace-mt1}
 \end{figure}

\begin{exa}
\cref{fig:trace-mt1} shows some trace $T_1$ resulting
from execution of the program in \cref{fig:real}.
We assume that the main thread always has thread id~$\thread{1}$.
For ease of reading, we use a tabular notation for traces
where each row contains exactly one event, in the column of the executing thread.
The textual order (from top to bottom) reflects the observed temporal order of events.
We write $e_i$ to refer to the $i$th in $T_1$.
For example, $e_5 = (5,\thread{2},\unlockE{\LKA})$.
In more compact form, trace $T_1$ can be written as $[e_1,...,e_{11}]$.
\end{exa}

Traces are well formed and respect the program's semantics.
For example, a lock cannot be acquired twice etc.
In the below, we write $l,l'$ to refer to lock events and $u,u'$ to refer
to unlock events.
We write $e \in T$ to indicate that $T=[e_1,...,e_n]$ where $e = e_k$ for some $k$,
defining also that $\pos{T}{e} = k$.
For events $e, f \in T$ we define trace order: $e \TrLt{T} f$ if $\pos{T}{e} < \pos{T}{f}$.
We write $e \TrLtEq{T} f$ if $e \TrLt{T} f$ or $e=f$.

\begin{defi}[Well Formedness]
\label{d:WF}
  Trace $T$ is \emph{well formed} if all the following conditions are satisfied:
  \begin{description}
    \item[\cond{WF-Acq}]
      For all $l = (\thdSym,\lockE{m}), l' = (\thdSym',\lockE{m}) \in T$ where $l \TrLt{T} l'$, there exists $u = (\thdSym,\unlockE{m}) \in T$ such that $l \TrLt{T} u \TrLt{T} l'$.

   \item[\cond{WF-Rel}]
      For all $u = (\thdSym,\unlockE{m}) \in T$ there exists $l = (\thdSym,\lockE{m}) \in T$ such that $l \TrLt{T} u$ and there is no $u' = (\thdSym',\unlockE{m}) \in T$ with $l \TrLt{T} u' \TrLt{T} u$.

  \item[\cond{WF-Fork1}]
    For all $\thdSym \ne \thread{1}$ there exists at most one $e = (\thdSymB,\forkE{t}) \in T$ and
    for all $e = (\thdSymB,\forkE{\thdSym}) \in T$ we have that $\thdSym \ne \thread{1}$.

  \item[\cond{WF-Fork2}]
    For all $e = (\thdSym,op) \in T$ where $\thdSym \ne \thread{1}$
    there exists $f = (\thdSymB,\forkE{\thdSym}) \in T$
    where $f \TrLt{T} e$.

  \item[\cond{WF-Join1}]
    For all $e = (\thdSymB,\joinE{\thdSym}) \in T$ we have $\thdSymB \ne \thdSym$ and there exists
    $f=(\thdSym,op) \in T$.

  \item[\cond{WF-Join2}]
    For all $e = (\thdSymB,\joinE{\thdSym}) \in T$ and all $f = (\thdSym,op) \in T$
    we have $f \TrLt{T} e$.
  \end{description}
\end{defi}

\cond{WF-Acq} states that a previously acquired lock can only be acquired after it has been released.
Similarly, \cond{WF-Rel} states that a lock can only be released after it has been acquired but not yet released.
Both conditions require that the operation that
releases a lock must occur in the same thread
as the operation that acquired the lock.

\cond{WF-Fork1} states that a thread can be created at most once.
\cond{WF-Fork2} states that each thread except the main thread is preceded by a fork event.
\cond{WF-Join1} states that a join operation $\joinE{t}$ must occur
in a thread distinct from $t$ and there must be operations that are executed
in thread~$t$.
\cond{WF-Join2} states that all events from a joined thread must appear before the join event.

\begin{exa}
\label{ex:well-ill-formed}
  Trace~$T_1$ is well-formed.
  Consider the following traces where $e_i$ refer to the events in \cref{fig:trace-mt1}:
  \bda{lcl}
     T_2 & =  & [e_1,e_6,e_7,e_8,e_9,e_{10},e_{11},e_2,e_3,e_4,e_5]
     \\  T_3 & = & [e_1,e_6,e_7,e_8,e_9,e_{10},e_{11},e_2]
     \\ T_4 & = & [e_1,e_7,e_8,e_{10},e_6,e_{11}]
     \\ T_5 & = & [e_1,e_2,e_6]
     \\ T_6 & = & [e_2,e_3,e_4,e_5]

   \eda
  Traces~$T_2$, $T_3$ and $T_4$ are well-formed as well.
  Traces $T_5$ and $T_6$ are ill-formed.
  Trace $T_5$ violates \cond{WF-Acq} and trace $T_6$ violates \cond{WF-Fork2}.
\end{exa}

A trace represents one possible interleaving of concurrent events.
In theory, there can be as many interleavings as there are (well-formed) reorderings of
the original trace. However, not all trace reorderings are feasible
in the sense that they could be reproduced by executing the program with a different schedule.

\begin{exa}
Traces~$T_2$, $T_3$ and $T_4$ are well-formed reorderings of a subset of the events in $T_1$.
Traces~$T_2$ and~$T_3$ represent executions where a different schedule is taken compared to~$T_1$.
In case of trace $T_2$, the  operations in thread~$\thread{2}$ are executed last.
Trace $T_3$ is a variant of $T_2$ where we assume that not all operations
in thread~$\thread{2}$ are executed because
the main thread~$\thread{1}$ terminates immediately after $e_{11}$ and
thus all other threads are terminated as well.
Trace $T_4$ however results from completely different program run because
the order among events in a thread has changed.
In fact, we can argue there is no program run of the example in~\cref{fig:real}
that leads to trace~$T_4$.
\end{exa}

We characterize all trace reorderings of trace~$T$
that represent alternative executions under a different schedule
in terms of correctly reordered prefixes.
Reorderings must be correct in the sense that the reordered trace is well-formed
and maintains order among events per thread as in the original trace~$T$.
Reorderings do not have to cover all events, so we consider reordered prefixes of traces.
We make use of the following helper functions and definitions.

The \emph{projection} of~$T$ onto thread~$\thdSym$, denoted $\projThd{T}{\thdSym}$, restricts~$T$ to events~$e$ where $e=(\thdSym,op)$.
We say trace $T$ \emph{prefixes} trace $T'$ if $T' = T \conc T''$ for some trace $T''$.
We define $\thds{T} = \{ \thdSym \mid \exists e = (\thdSym,op)  \in T \}$.
For example, $\projThd{T_3}{\thread{2}} = [e_2]$ and
$\evts{T_3} = \{e_1,e_6,e_7,e_8,e_9,e_{10},e_{11},e_2\}$
and $\thds{T_1} = \{ \thread{1},\thread{2},\thread{3} \}$.

\begin{defi}[Correctly Reordered Prefix]
  \label{def:crp}
  Let $T$ be a well-formed trace.
  Then, trace ~$T'$ is a \emph{correctly reordered prefix} of~$T$ if the following
  conditions are satisfied:
  \begin{description}
    \item[\cond{CRP-WF}] $T'$ is well formed. 
    \item[\cond{CRP-PO}] For every~$\thdSym \in \thds{T'}$, $\projThd{T'}{\thdSym}$ prefixes~$\projThd{T}{\thdSym}$.
  \end{description}
  We write $\crp{T}$ to denote the set of correctly reordered prefixes of~$T$.
\end{defi}

\cond{CRP-PO} states that the order of events in a thread must be maintained.
\cond{CRP-WF} states well-formedness and thus guarantees the conditions in~Definition~\ref{d:WF}. For example, $T_2,T_3 \in \crp{T_1}$ and $T_4,T_5,T_6 \not\in \crp{T_1}$.

\section{Critical Section and Lock Set}
\label{sec:cs-ls}

We consider a specific program run represented by some well-formed trace $T$.
In trace $T$ we wish to identify critical sections that protect events.
The entry and exit points of a critical section are identified by
matching lock and unlock events~$l$ and $u$
where $l$ and $u$ operate on the same lock and
are nearest to each other.
By well-formedness events~$l$ and $u$ must belong to the same thread.
Hence, pairs of matching lock and unlock events can be easily
computed via a linear scan of the trace and we
write $\exit{T}{l}$ to denote the matching unlock event for lock event~$l$.
A critical section might be 'open' in the sense that
a lock event~$u$ might not have a matching unlock event~$l$.
In such cases we assume that $\exit{T}{l}$ equals~$\bot$.
The definition below formally defines $\exit{T}{l}$.

\begin{defi}[Entry and Exit Points of Critical Sections]
  Let $T$ be a well-formed trace with lock event $l=(\thdSym,\lockE{m}) \in T$.
  We refer to $l$ as the \emph{entry point} of a critical section
  on lock $m$. The \emph{exit point} of $l$ in trace~$T$
  is obtained via the function $\exit{T}{l}$ defined below.

  \bda{lcll}
  \exit{e:T}{l} & = &
   \left \{
   \ba{ll}
     \exitH{T}{\thdSym,m} & \mbox{if $e=l$}
   \\  \exit{T}{l} & \mbox{otherwise}
   \ea
    \right.
  \\
  \\
  \exitH{[]}{\thdSym,m} & = & \bot
  \\
  \exitH{e:T}{\thdSym,m} & = &
  \left \{
  \ba{ll}
     e &
  \mbox{if $e=(\thdSym,\unlockE{m})$}
  \\
  \exitH{T}{\thdSym,m} & \mbox{otherwise}
  \ea
  \right.
  \eda
  We assume that $\exit{T}{}$ and $\exitH{T}{}$ describe a family
  of functions indexed by some trace~$T$.
  Function $\exit{T}{l}$ scans the trace for entry point~$l$
  and then returns the exit point via helper function $\exitH{T}{\thdSym,m}$.
\end{defi}

\begin{exa}
  For trace~$T_1$ from~\cref{fig:trace-mt1}
  we find that
  $\exit{T_1}{e_6}=e_{11}$,
  $\exit{T_1}{e_8}=e_9$,
  $\exit{T_1}{e_2}=e_5$ and
  $\exit{T_1}{e_3}=e_4$.
  For trace $T_3 = [e_1,e_6,e_7,e_8,e_9,e_{10},e_{11},e_2]$
  from Example~\ref{ex:well-ill-formed} we find
  that $\exit{T_3}{e_6}=e_{11}$
  and $\exit{T_3}{e_8}=e_9$ and $\exit{T_3}{e_2} = \bot$.
  The last case shows that a lock event, here~$e_2$, might
  lack a matching unlock event.
\end{exa}

The exit point for each entry point is stable under trace reorderings.

\begin{lem}[Stability of Exit Points]
\label{lem:stable-exit}
Let $T$ be a well-formed trace with lock event $l = (t,\lockE{m}) \in T$.
If $\exit{T}{l} \ne \bot$
then for any reordering $T' \in \crp{T}$ where $\exit{T}{l} \in T'$
we have that $\exit{T}{l} = \exit{T'}{l}$.
If $\exit{T}{l} = \bot$
then for any reordering $T' \in \crp{T}$
we have that $\exit{T'}{l} = \bot$.
\end{lem}
\begin{proof}
  By construction we find that
  (1) $\exit{T}{l} = \exit{\projThd{T}{\thdSym}}{l}$ and
  (2) for any $T' \in \crp{T}$ we have that $\projThd{T'}{\thdSym}$
  is a prefix of $\projThd{T}{\thdSym}$.
  Hence, if $\exit{T}{l} = \bot$ we immediately find
  that $\exit{T'}{l} = \bot$ for any $T' \in \crp{T}$.

  Consider the case that $\exit{T}{l} \ne \bot$.
  Consider $T' \in \crp{T}$ where $\exit{T}{l} \in T'$.
  Then, from (1) and (2) we can immediately follow that
  $\exit{T}{l} = \exit{T'}{l}$.
\end{proof}

We can thus argue that critical sections are uniquely identified by
their entry and exit points.
An event is part of a critical
section if the event is enclosed by the entry
and its corresponding exit point (if defined).
In such a situation we say that the event is protected by the lock
that is acquired at the entry point.
The set of locks that protect an event is commonly
referred to as the \emph{lock set}.
The important question is how to check if an event is 'enclosed'
by another event.

Next, we review the existing per-thread definition of a critical section and
the resulting lock set construction
where the check for enclosed events only considers
events that are in the same thread as the exit point.
We point out the semantic incompleteness of this assumption
and introduce a more general trace-based characterization
that omits the per-thread restriction and
is shown to be semantically complete.
To establish correctness and completeness, we employ
the following characterization of \emph{lock protection}
that  does not refer to critical sections explicitly.

\begin{defi}[Lock Protection]
\label{def:lock-protection}
  Let $T$ be a well-formed trace with event $e \in T$.
  We say that \emph{$e$ is protected by lock $m$} if
  there exists lock event $l=(\thdSym,\lockE{m}) \in T$
  such that the following two conditions are satisfied:
  \begin{description}
  \item[\cond{LP1}]
      For all $T' \in \crp{T}$ where
       $l, e \in T'$ we have that $l \TrLt{T'} e$.
     \item[\cond{LP2}]
         There is no unlock event $u = (\thdSym,\unlockE{m}) \in T$
         such that for some $T' \in \crp{T}$ where
       $l,u,e \in T'$ we have that $l \TrLt{T'} u \TrLt{T'} e$.
  \end{description}
\end{defi}

Case \cond{LP1} guarantees that under all schedules
the event~$e$ is preceded by the event that acquires lock~$m$.
Case \cond{LP2} guarantees that the lock~$m$ is not released in any
of these schedules.

\subsection{Per-Thread Definition}
The following definition of a critical section and lock set is commonly employed in the literature.

\begin{defi}[Per-Thread Critical Section]
  Let $T$ be a well-formed trace with
  lock event $l=(\thdSym,\lockE{m})\in T$.
  We say that $e \in T$ is
  in the \emph{per-thread critical section guarded by $l$ in thread $\thdSym$},
  written $e \in \PerThreadCS{T}{\thdSym}{l}$,
  if one of the following two conditions apply:

  \begin{description}
  \item[\cond{OPEN-PT}]
     $\exit{T}{l} = \bot$ and
     $\projThd{T}{\thdSym} = T_1 \conc [l] \conc T_2 $
    for some traces~$T_1$ and $T_2$ where~$e \in T_2$.
  \item[\cond{CLOSED-PT}]
         $\exit{T}{l} \ne \bot$ and
        $\projThd{T}{\thdSym} = T_1 \conc [l] \conc T_2 \conc [\exit{T}{l}] \conc T_3$
    for some traces~$T_1$, $T_2$ and $T_3$ where  $e \in T_2$.
  \end{description}
\end{defi}
The above states that an event $e$ is in the per-thread critical section
guarded by~$l$ if $e$ and $l$ belong to the same thread
and $e$ occurs after $l$ and
if $\exit{T}{l}$ is defined, $e$ occurs before $\exit{T}{l}$ in the trace.
This means that the check if $e$ is enclosed
by some critical section guarded by~$l$ only considers
events that are in the same thread as~$l$.

\begin{defi}[Per-Thread Locks Held]
Let $T$ be a well-formed trace with an event $e =(\thdSym,op) \in T$.
  The \emph{per-thread lock set for an event $e\in T$ in thread $\thdSym$} is defined by
  $\PerThreadLH{T}{\thdSym}{e} =
  \{ m \mid \exists l = (\thdSym,\lockE{m}) \in T.
        e \in \PerThreadCS{T}{\thdSym}{l} \}$.
\end{defi}

The superscript $\thdSym$ in $\PerThreadCS{T}{\thdSym}{l}$
and $\PerThreadLH{T}{\thdSym}{e}$ highlights that the above definitions
are restricted to events that are in the same thread~$\thdSym$.

\begin{exa}
  For trace~$T_1$ in \cref{fig:trace-mt1}
  and events $e_3 = (\thread{2}, \lockE{m_2}), e_8 = (\thread{3}, \lockE{m_2}) \in T_1$
  we find
  that $\PerThreadLH{T_1}{\thread{2}}{e_3} = \{ m_1 \}$
  and $\PerThreadLH{T_1}{\thread{3}}{e_8} = \{ \}$.
\end{exa}

Each lock in the per-thread lock set protects the associated event
as shown by the following result.

\begin{lem}[Semantic Correctness of Per-Thread Lock Set]
  \label{lem:per-thread-correct}
  Let $T$ be a well-formed trace with some event $e=(\thdSym,op) \in T$
  where $m \in \PerThreadLH{T}{\thdSym}{e}$.
  Then, $e$ is protected by lock $m$.
\end{lem}
\begin{proof}
  By assumption, there exists $l=(\thdSym,\lockE{m}) \in T$
  and $e \in \PerThreadCS{T}{\thdSym}{l}$.

  \mbox{} \\
  \noindent
  {\bf Case $\exit{T}{l} \ne \bot$:}
  Then, (1) $\projThd{T}{\thdSym} = T_1 \conc [l] \conc T_2 \conc [\exit{T}{l}] \conc T_3$
  for some traces~$T_1$, $T_2$ and $T_3$ where~$e \in T_2$.

  Consider $T' \in \crp{T}$ where~$l,e \in T'$.
  Each $T' \in \crp{T}$ maintains the program order condition \cond{CRP-PO}.
  Hence, from (1) we derive that $l \TrLt{T'} e$.
  This establishes condition~$\cond{LP1}$.

  We establish the condition $\cond{LP2}$ by contradiction.
  Suppose there exists some unlock event $u = (\thdSym,\unlockE{m}) \in T$
  such that for some $T' \in \crp{T}$ where~$l,u,e \in T'$ we have that (2) $l \TrLt{T'} u \TrLt{T'} e$.
  We choose the unlock event $u$ that is nearest to~$l$
  with no lock event on the same lock in between.
  The well-formedness assumption guarantees that such an unlock event exists.
  Hence, $\exit{T'}{l} = u$.
  By Lemma~\ref{lem:stable-exit} we find that $\exit{T}{l} = u$.

  By assumption (1), event $e$ appears before $\exit{T}{l}$
  but from (2) we obtain that $e$ appears after $\exit{T}{l}$.
  Events $e$ and $\exit{T}{l}$ belong to the same thread.
  Hence, the program order condition~\cond{CRP-PO} is violated
  and we achieve the desired contradiction.

  \mbox{} \\
  \noindent
  {\bf Case $\exit{T}{l} = \bot$:}
  We find that (3) $\projThd{T}{\thdSym} = T_1 \conc [l] \conc T_2 $
  for some traces~$T_1$ and $T_2$ where~$e \in T_2$.

  Consider $T' \in \crp{T}$ where $l,e \in T'$.
  Each $T' \in \crp{T}$ maintains the program order condition \cond{CRP-PO}.
  Hence, from (3) we derive that $l \TrLt{T'} e$.
  This establishes condition~$\cond{LP1}$.

    We establish the condition $\cond{LP2}$ by contradiction.
  Suppose there exists some unlock event $u = (\thdSym,\unlockE{m}) \in T$
  such that for some $T' \in \crp{T}$ where
  $l,u,e \in T'$ we have that (2) $l \TrLt{T'} u \TrLt{T'} e$.
  We choose the unlock event $u$ that is nearest to~$l$
  with no lock event on the same lock in between.
  The well-formedness assumption guarantees that such an unlock event exists.
  Hence, $\exit{T'}{l} = u$.
  By Lemma~\ref{lem:stable-exit} we find that $\exit{T}{l} = u$.
  This contradicts the assumption that $\exit{T}{l} = \bot$.

  For both cases we can establish conditions $\cond{LP1}$ and $\cond{LP2}$.
  Thus, we can conclude that~$e$ is protected by lock~$m$.
\end{proof}

However, the per-thread lock set definition is semantically incomplete.

\begin{exa}
\label{exa:per-thread-incomplete}
Consider trace $T_1$ in \cref{fig:trace-mt1} where
$e_{11} = \exit{T_1}{e_6}$.
Due to the fork-join dependency we find
that for all reorderings $T' \in \crp{T_1}$ where $e_{11} \in T'$
we have that event $e_8$ is always surrounded by events $e_6$ and $e_{11}$.
That is, $e_6 \TrLt{T'} e_8 \TrLt{T'} e_{11}$.
This means that $e_8$ is protected by lock~$m$, however,
$m \not\in \PerThreadLH{T_1}{e_8}$ because $e_8$ arises in
a different thread than~$e_6$.
\end{exa}

\subsection{Trace-Based Characterization}

To guarantee that the lock set is semantically complete,
we give a trace-based characterization that lifts the per-thread restriction
when checking if an event is enclosed by another event.
For this purpose, we introduce the following ordering relation.
Let $e,f$ be two events in trace $T$.
We say event $f$ \emph{must be preceded}
by event~$e$ \emph{under all trace reorderings of $T$},
written $e \crpLt{T} f$, if $\forall T' \in \crp{T}$ where $f \in T'$
we have that $e\in T'$ and $e \TrLt{T'} f$.

\begin{defi}[Critical Section]
\label{def:cs}
  Let $T$ be a well-formed trace with lock event $l=(\thdSym,\lockE{m}) \in T$.
  We say that $e \in T$ is
  in the \emph{critical section guarded by $l$},
  written $e \in \CS{T}{l}$,
  if  one of the following two conditions apply:
  \begin{description}
  \item[\cond{OPEN}]
     $\exit{T}{l} = \bot$ and $l \crpLt{T} e$.
  \item[\cond{CLOSED}]
         $\exit{T}{l} \ne \bot$ and $l \crpLt{T} e \crpLt{T} \exit{T}{l}$.
 \end{description}
\end{defi}
The above states that an event $e$ is part of a critical section
guarded by~$l$ if the event is $e$ is enclosed
by the critical section's entry point~$l$
regardless of the schedule of events.
If the exit point~$\exit{T}{l}$ exit, $e$ is also enclosed by~$\exit{T}{l}$.
Importantly, event $e$ does not need to arise in the same thread as~$l$.
Thus, the above definition can deal with critical sections
that cover multiple threads as seen by the example in~\cref{fig:trace-mt1}.

\begin{defi}[Locks Held]
Let $T$ be a well-formed trace with an event $e =(\thdSym,op) \in T$.
  The \emph{lock set for an event $e\in T$ in thread $\thdSym$} is defined by
  $\LH{T}{e} =
  \{ m \mid \exists l=(\thdSym,\lockE{m})  \in T.
        e \in \CS{T}{l} \}$.
\end{defi}

\begin{lem}[Semantic Correctness and Completeness of Lock Set]
\label{lem:trace-based-complete}
  Let $T$ be a well-formed trace with some event $e=(\thdSym,op) \in T$.
  Then, $m \in \LH{T}{e}$ iff
         $e$ is protected by lock $m$.
\end{lem}
\begin{proof}

  \mbox{} \\
  \noindent ``$\Rightarrow$'':  We consider first the direction from left to right. We have that
  $m \in \LH{T}{e}$
  iff
  there exists $l=(\thdSym,\lockE{m})  \in T$ and $e \in \CS{T}{l}$.

  By assumption $l \crpLt{T} e$ and thus
  we can establish condition~$\cond{LP1}$.

    \mbox{} \\
  \noindent
  {\bf Case $\exit{T}{l} \ne \bot$:}
  We find that (1) $e \crpLtEq{T} \exit{T}{l}$.
  We establish condition~$\cond{LP2}$ by contradiction
  via similar reasoning as in the proof of Lemma~\ref{lem:per-thread-correct}.

  Suppose there exists some unlock event $u = (\thdSym,\unlockE{m}) \in T$
  such that for some $T' \in \crp{T}$ where
  $l,u,e \in T'$ we have that (2) $l \TrLt{T'} u \TrLt{T'} e$.
  We choose the unlock event $u$ that is nearest to~$l$
  with no lock event on the same lock in between.
  The well-formedness assumption guarantees that such an unlock event exists.
  Hence, $\exit{T'}{l} = u$.
  By Lemma~\ref{lem:stable-exit} we find that $\exit{T}{l} = u$.

  From (1) we derive that the event $e$ appears before $\exit{T}{l}$
  but from (2) we obtain that $e$ appears after $\exit{T}{l}$.
  Events $e$ and $\exit{T}{l}$ belong to the same thread.
  Hence, the program order condition~\cond{CRP-PO} is violated
  and we achieve the desired contradiction.

  \mbox{} \\
  \noindent
      {\bf Case $\exit{T}{l} = \bot$:}
      The reasoning that leads to a contradiction is exactly
      the same as in the proof of Lemma~\ref{lem:per-thread-correct}.

\mbox{} \\
\noindent ``$\Leftarrow$'':
We consider the direction from right to left.

From condition~$\cond{LP1}$ we can  immediately
derive that $l \crpLt{T} e$.
Consider the case that $\exit{T}{l} \ne \bot$.
It remains to show that $e \crpLt{T} \exit{T}{l}$.
Assume the contrary.

This means there exists $T' \in \crp{T}$
where $\exit{T}{l} \TrLt{T'} e$.
By construction we also find that $l \TrLt{T'} \exit{T}{l}$.
This contradicts condition~$\cond{LP2}$ and we are done.
\end{proof}

An immediate consequence is that
the per-thread lock set construction is a
strict subset of the trace-based lock set construction.

From Lemma~\ref{lem:per-thread-correct} and Lemma~\ref{lem:trace-based-complete}
we derive that $\PerThreadLH{T}{\thdSym}{e} \subseteq \LH{T}{e}$
for any $e \in T$.
The subset relation is strict as shown
by Example~\ref{exa:per-thread-incomplete}.

\section{Related Work}
\label{sec:related-work}

The notion of a critical section goes back to Dijkstra~\cite{DBLP:journals/cacm/Dijkstra65}. It is hard to find a mathematical definition.
We review earlier works in the context of program analysis
that rely on the notion of critical section and lock set.

\paragraph{Deadlock detection}

The use of lock set has emerged in the context of deadlock detection.
For example, see the earlier dynamic deadlock detection works by
Havelund~\cite{10.5555/645880.672085} and
Harrow~\cite{DBLP:conf/spin/Harrow00}.
There are numerous follow-up works, e.g.,
see~\cite{conf/hvc/BensalemH05,conf/pldi/JoshiPSN09,10.1145/3377811.3380367,Samak:2014:TDD:2692916.2555262,conf/ase/ZhouSLCL17,6718069,conf/oopsla/KalhaugeP18,conf/fse/CaiYWQP21,conf/pldi/TuncMPV23}.

Our results apply independently of how traces are obtained (e.g., via instrumentation, symbolic execution, or CFG reachability).
Static analysis works typically  consider the reachability graph.
There are numerous works that employ the lock set method
for static deadlock detection,
e.g., see~\cite{10.1145/2970276.2970309,
  10.1109/APSEC.2008.68,10.1109/ICSE.2009.5070538,10.1145/3540250.3549110,10.1145/2970276.2970309}.

\paragraph{Data race detection}

The notion of critical section and lock set is also employed
by works that cover data race detection.
For example, consider lock set based data race method~\cite{Savage:1997:EDD:265924.265927,serebryany2009threadsanitizer,xie2013acculock,yu2016simplelock+}
and
works that order conflicting critical sections~\cite{Smaragdakis:2012:SPR:2103621.2103702,DBLP:journals/corr/KiniM017,Roemer:2018:HUS:3296979.3192385-obsolete,10.1145/3360605,conf/mplr/SulzmannS20}.

In all these works, critical sections and lock sets are implicitly defined per thread.

\section{Conclusion}
\label{sec:conclusion}

Pretty much all prior works assume per-thread lock sets and critical
sections restricted to a single thread.
This assumption does not hold for general C/Pthread executions.
Critical sections that cover
multiple threads exist, can be observed in simple programs and invalidate the standard foundation.
Similar phenomena can occur in other thread-and-monitor models such as Java.
Correcting this semantic gap is important because lock sets form the foundation of numerous analyses for deadlock and data race detection.
We consider a minimal trace model for C/Pthread programs and
give corrected definitions based on a trace-based characterization of critical sections and lock
sets to deal with programs where locks protect events across thread boundaries.
Extending the semantics with memory operations and additional
synchronization primitives such as condition variables
introduces orthogonal challenges and is left for future work.



\newcommand{\etalchar}[1]{$^{#1}$}

\end{document}